\documentclass[runningheads]{llncs}

\usepackage[utf8]{inputenc}
\usepackage{amsmath}
\usepackage{amssymb}
\usepackage{graphicx}
\usepackage{tabularx}
\usepackage{array}
\usepackage{pifont}
\usepackage{xspace}
\usepackage{todonotes}
\usepackage{paralist}
\usepackage{thm-restate}
\usepackage{subfigure}
\usepackage{complexity}
\usepackage[pdfpagelabels,colorlinks,citecolor=blue,linkcolor=blue,urlcolor=blue]{hyperref}
\usepackage[english]{babel}
\usepackage{amsopn}
\usepackage{latexsym}
\usepackage{multirow}
\usepackage{multicol}
\usepackage{booktabs}
\usepackage[capitalise]{cleveref}
\usepackage{color, colortbl}
\usepackage{booktabs}

\renewcommand{\orcidID}[1]{\href{https://orcid.org/#1}{\includegraphics[scale=.03]{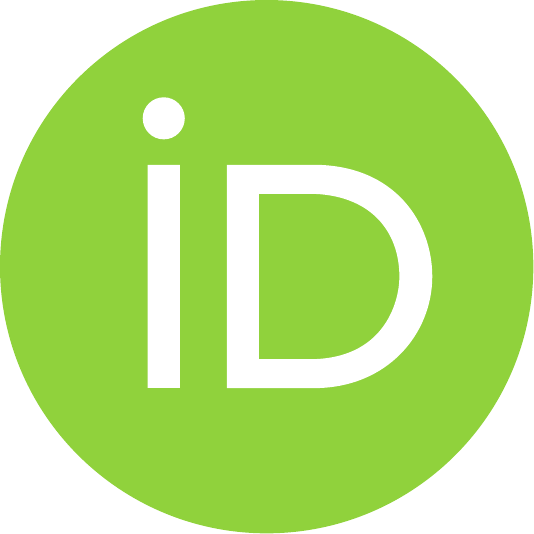}}}

\spnewtheorem{clm}{Claim}{\bfseries}{\rmfamily}

\definecolor{lightcyan}{rgb}{0.88,1,1}
\definecolor{antiquewhite}{rgb}{0.98, 0.92, 0.84}
\DeclareMathOperator\indeg{indeg}
\DeclareMathOperator\outdeg{outdeg}

\newcounter{casecounter}
\newcounter{subcasecounter}
\newcounter{subsubcasecounter}
\makeatletter
\newcommand{\ccase}[2][]{%
	\~counter{casecounter}%
	\setcounter{subcasecounter}{0}%
	\protected@write \@auxout {}{\string \newlabel {#2}{{#1\thecasecounter}{\thepage}{#1\thecasecounter}{#2}{}} }%
	\hypertarget{#2}{\noindent\textbf{Case #1\thecasecounter.}}
}

\newcommand{\subcase}[2][]{%
	\`counter{subcasecounter}%
	\setcounter{subsubcasecounter}{0}%
	\protected@write \@auxout {}{\string \newlabel {#2}{{#1\thecasecounter.\thesubcasecounter}{\thepage}{#1\thecasecounter.\thesubcasecounter}{#2}{}} }%
	\hypertarget{#2}{\noindent\textbf{Case #1\thecasecounter.\thesubcasecounter.}}
}

\newcommand{\subsubcase}[2][]{%
	\stepcounter{subsubcasecounter}%
	\protected@write \@auxout {}{\string \newlabel {#2}{{#1\thecasecounter.\thesubcasecounter.\thesubsubcasecounter}{\thepage}{#1\thecasecounter.\thesubcasecounter.\thesubsubcasecounter}{#2}{}} }%
	\hypertarget{#2}{\noindent\textbf{Case #1\thecasecounter.\thesubcasecounter.\thesubsubcasecounter.}}
}
\makeatother

\newcommand{\skel}{\mathrm{skel}\xspace}

\definecolor{lightgreen}{rgb}{0.6, 0.98, 0.6}

\newcommand{\ortho}{\textsc{$b$-Orthogonal Planarity}\xspace}
\newcommand{\rpt}{\textsc{Rectlinear Planarity}\xspace}

\begin{document}
	\title{On the Parameterized Complexity of Bend-Minimum Orthogonal Planarity \thanks{Research partially supported by: (i) University of Perugia, Ricerca Base 2021, Proj. ``AIDMIX — Artificial Intelligence for
			Decision Making: Methods for Interpretability and eXplainability''; (ii) MUR PRIN Proj. 2022TS4Y3N - ``EXPAND: scalable algorithms for EXPloratory Analyses of heterogeneous and dynamic Networked Data''; (iii) MUR PRIN Proj. 2022ME9Z78 - ``NextGRAAL: Next-generation algorithms for constrained GRAph visuALization''}}
	 \author{Emilio Di Giacomo\inst{1}\orcidID{0000-0002-9794-1928},Walter Didimo\inst{1}\orcidID{0000-0002-4379-6059},
     Giuseppe Liotta\inst{1}\orcidID{0000-0002-2886-9694},
     \\Fabrizio Montecchiani\inst{1}\orcidID{0000-0002-0543-8912},
    Giacomo Ortali\inst{1}\orcidID{0000-0002-4481-698X}$^c$
	 }
	
	\date{}
	
	 \institute{
		Universit\`a degli Studi di Perugia, Italy\\
		\email {\{name.surname\}@unipg.it}
	 }

	\maketitle
	
	%
\begin{abstract}
Computing planar orthogonal drawings with the minimum number of bends is one of the most relevant topics in Graph Drawing. The problem is known to be \NP-hard, even when we want to test the existence of a rectilinear planar drawing, i.e., an orthogonal drawing without bends (Garg and Tamassia, 2001). From the parameterized complexity perspective, the problem is fixed-parameter tractable when parameterized by the sum of three parameters: the number of bends, the number of vertices of degree at most two, and the treewidth of the input graph (Di Giacomo et al., 2022). We improve this last result by showing that the problem remains fixed-parameter tractable  when parameterized only by the number of vertices of degree at most two plus the number of bends. As a consequence, rectilinear planarity testing lies in \FPT~parameterized by the number of  vertices of degree at most two.

\end{abstract}

\section{Introduction}
Orthogonal drawings of planar graphs represent a foundational graph drawing paradigm~\cite{DBLP:books/ph/BattistaETT99,DBLP:journals/siamcomp/BattistaLV98,DBLP:journals/jgaa/RahmanNN99} and an active field of research~\cite{DBLP:journals/jcss/GiacomoLM22,DBLP:conf/gd/DidimoKLO22,DBLP:conf/soda/DidimoLOP20,DBLP:journals/comgeo/Frati22}. Given a planar 4-graph $G$, i.e., a planar graph of vertex-degree at most 4, a \emph{planar orthogonal drawing} $\Gamma$ of $G$ is a planar drawing that maps the vertices of $G$ to distinct points of the plane and each edge of $G$ to a chain of horizontal and vertical segments connecting the two endpoints of the edge. A \emph{bend} of $\Gamma$ is the meeting point of two consecutive segments along the same edge. Given a graph $G$ and $b \in \mathbb{N}$, the \ortho problem asks for the existence of an orthogonal drawing of $G$ with at most $b$ bends in total. While every $n$-vertex planar 4-graph admits an orthogonal drawing with $O(n)$ bends~\cite{DBLP:journals/comgeo/BiedlK98},  \ortho  is well-known to be \NP-complete already when $b=0$~\cite{DBLP:journals/siamcomp/GargT01}. In this case, the \ortho problem is also known in the literature as \rpt. It is even \NP-hard to approximate the minimum number of bends in an orthogonal drawing with an $O(n^{1-\varepsilon})$ error for any $\varepsilon > 0$~\cite{DBLP:journals/siamcomp/GargT01}.

From a parameterized complexity perspective, the hardness of \rpt rules out the possibility of finding a tractable solution for  \ortho  parameterized solely by the number of bends. On the positive side, a recent result of Di Giacomo et al.~\cite{DBLP:journals/jcss/GiacomoLM22} shows that \ortho lies in the \textsf{XP} class when parameterized by the treewidth $\mathsf{tw}$ of the input graph. In fact, the result in~\cite{DBLP:journals/jcss/GiacomoLM22} is based on a fixed-parameter tractable algorithm in the parameter $b+k+\mathsf{tw}$, where $k$ is the number of vertices of degree at most two. A natural question is thus whether a subset of these three parameters suffices to establish tractability, as asked in~\cite{DBLP:journals/jcss/GiacomoLM22}. We also remark that, in a recent Dagstuhl seminar on parameterized complexity in graph drawing~\cite[pag. 94]{DBLP:journals/dagstuhl-reports/GanianMNZ21},  understanding whether \rpt is in \FPT~parameterized by only one of $k$ and $\mathsf{tw}$ was identified as a~prominent~challenge. 

\medskip\noindent\textbf{Contribution.} 
We address the questions posed in~\cite{DBLP:journals/jcss/GiacomoLM22,DBLP:journals/dagstuhl-reports/GanianMNZ21} and prove that \ortho is fixed-parameter tractable parameterized by $b+k$, hence dropping the dependency on $\mathsf{tw}$. Consequently, \rpt is in \FPT~when parameterized by $k$ alone.

From a technical point of view, the algorithm by Di Giacomo et al.~\cite{DBLP:journals/jcss/GiacomoLM22} exploits treewidth to apply dynamic programming. The main crux of that algorithm lies on two main ingredients. First, a fixed treewidth (and fixed vertex-degree) bounds the number of interesting embeddings in which different faces have vertices in the current bag. Second, a classic result by Tamassia~\cite{DBLP:journals/siamcomp/Tamassia87} states that the existence of an orthogonal drawing is equivalent to the existence of a combinatorial representation describing the angles at the vertices and the bends along the edges. In this respect, a fixed number of bends and of vertices of degree at most two bounds the possible combinatorial descriptions of a face, and hence makes the definition of small records feasible for the sake of dynamic programming. To avoid the dependency on treewidth, in our work we exploit additional insight into the structure of the problem. First, we adopt SPQ$^*$R-trees to decompose the input graph and to keep track of its possible embeddings. While it was already known how to deal with S- and P-nodes efficiently by exploiting dynamic programming on SPQ$^*$R-trees~\cite{DBLP:journals/siamcomp/BattistaLV98,Didimo2023}, the main obstacle is represented by R-nodes, i.e., the rigid components of the graph. To overcome this obstacle, we first prove that the number of children of an R-node depends on the number of degree-2 vertices plus the number of bends. Then, we 
carefully combine flow-network based techniques~\cite{DBLP:journals/siamcomp/Tamassia87} for planar embedded graphs with the notion of orthogonal spirality~\cite{DBLP:journals/siamcomp/BattistaLV98}, which measures how much a component of an orthogonal drawing is rolled-up. 

The most natural question that remains open is to settle the complexity of \ortho parameterized  by $k$ alone, which we conjecture to be at least W[1]-hard due to the fact that bends and low-degree vertices have interchangeable behaviours in terms of spirality.


\section{Preliminaries}\label{se:preliminaries}
We only consider planar 4-graphs, i.e., graphs with vertex-degree at most 4. A \emph{plane graph} is a planar graph with a fixed planar embedding, i.e., prescribed clockwise orderings of the edges around the vertices and a given external face. 

\smallskip\noindent{\bf Orthogonal  Representations.} 
A \emph{planar orthogonal representation}~$H$ of a planar graph $G$ describes the shape of a class of orthogonal drawings in terms of left/right bends along the edges and angles at the vertices. A drawing $\Gamma$ of~$H$ (i.e., a planar orthogonal drawing of $G$ that preserves $H$) can be computed in linear time~\cite{DBLP:journals/siamcomp/Tamassia87}. If~$H$ has no bends, it is a \emph{rectilinear representation}. Since we only deal with planar drawings, we just use the term rectilinear (orthogonal) representation in place of planar rectilinear (orthogonal) representation.


Let $G$ be a plane graph with $n$ vertices. In \cite{DBLP:journals/siamcomp/Tamassia87}, it is proved that a bend-minimum orthogonal representation of $G$ can be constructed in $O(n^2 \log n)$ time by finding a min-cost flow on a suitable flow-network constructed from the planar embedding of $G$.  
The feasible flows on this network correspond to the possible orthogonal representations of the graph, and the total cost of the flow equals the number of bends of the corresponding representation. 
Fixing the values of the flow along some arcs of the network, one can force desired vertex-angles or left/right bends along the edges in the corresponding orthogonal representation. 
The time complexity in \cite{DBLP:journals/siamcomp/Tamassia87} has been subsequently improved to $O(n^\frac{7}{4}\sqrt{\log n})$ \cite{DBLP:conf/gd/GargT96a} time, and later to $O(n^\frac{3}{2})$\cite{DBLP:journals/jgaa/CornelsenK12} if there is no constraints on the orthogonal representation and one can use a flow-network whose arcs all have infinite capacities. 
 
Let $H'$ be an orthogonal representation of a (not necessarily connected) subgraph $G'$ of $G$. An orthogonal representation $H$ of $G$ is \emph{$H'$-constrained} if its restriction to $G'$ equals $H'$. By the considerations above, the following holds. 

\begin{lemma}
\label{le:planebendmin}
Let $G$ be an $n$-vertex plane 4-graph and let $H'$ be an orthogonal representation of a subgraph $G'$ of $G$. There exists an $O(n^\frac{7}{4}\sqrt{\log n})$-time algorithm that computes an $H'$-constrained orthogonal representation of $G$ with the minimum number of bends, if it exists, or that rejects the instance otherwise. 
\end{lemma}

\smallskip\noindent{\bf SPQR-trees.}
Let $G$ be a biconnected planar graph. The \emph{SPQR-tree} $T$ of $G$, introduced in~\cite{DBLP:books/ph/BattistaETT99}, represents the decomposition of $G$ into its triconnected components~\cite{DBLP:journals/siamcomp/HopcroftT73}.  
Each triconnected component corresponds to a non-leaf node $\nu$ of~$T$; the triconnected component itself is called the \emph{skeleton} of $\nu$ and is denoted as $\skel(\nu)$. Node $\nu$ can be: $(i)$ an \emph{R-node}, if $\skel(\nu)$ is a triconnected graph; $(ii)$ an \emph{S-node}, if $\skel(\nu)$ is a simple cycle of length at least three; $(iii)$ a \emph{P-node}, if $\skel(\nu)$ is a bundle of at least three parallel edges.
A degree-1 node of $T$ is a \emph{Q-node} and represents a single edge of~$G$.
A \emph{real edge} (resp. \emph{virtual edge})  in $\skel(\nu)$ corresponds to a Q-node (resp., to an S-, P-, or R-node) adjacent to $\nu$ in $T$.
Neither two S- nor two P-nodes are adjacent in~$T$. The SPQR-tree of a biconnected graph can be computed in linear time~\cite{DBLP:books/ph/BattistaETT99,DBLP:conf/gd/GutwengerM00}.

Let $e$ be a designated edge of $G$, called the \emph{reference edge} of $G$, let $\rho$ be the Q-node of $T$ corresponding to $e$, and let $T$ be rooted at $\rho$. 
For any P-, S-, or R-node $\nu$ of $T$, $\skel(\nu)$ has a virtual edge, called \emph{reference edge} of $\skel(\nu)$ and of $\nu$, associated with a virtual edge in the skeleton of its parent. The reference edge of the root child of $T$ is the edge corresponding to $\rho$.
The rooted tree $T$ describes all planar embeddings of $G$ with its reference edge on the external face; they are obtained by combining the different planar embeddings of the skeletons of P- and R-nodes with their reference edges on the external face. Namely, for a P- or R-node $\nu$, denote by $\skel^-(\nu)$ the skeleton of $\nu$ without its reference edge.
If $\nu$ is a P-node, the embeddings of $\skel(\nu)$ are the different permutations of the edges of $\skel^-(\nu)$; If $\nu$ is an R-node, $\skel(\nu)$ has two possible embeddings, obtained by flipping $\skel^-(\nu)$  at its poles.
For every node $\nu \neq \rho$, the \emph{pertinent graph $G_\nu$ of $\nu$} is the subgraph of $G$ whose edges correspond to the Q-nodes in the subtree of $T$ rooted at $\nu$.  
We also say that $G_\nu$ is a \emph{component of $G$}.
The pertinent graph $G_\rho$ of the root $\rho$ coincides with the reference edge of $G$.
If $H$ is an orthogonal representation of $G$, its restriction $H_\nu$ to $G_\nu$ is an \emph{orthogonal component} of $H$.

\section{The Biconnected Case}
\label{se:bico}
In this section we assume that $G$ is a biconnected graph with $k$ degree-2 vertices. Similar to previous works~\cite{DBLP:journals/siamcomp/BattistaLV98,DBLP:conf/gd/DidimoKLO22,Didimo2023,DBLP:conf/isaac/DidimoL98,DBLP:journals/jcss/DidimoLP19}, our approach exploits a variant of SPQR-tree, called \emph{SPQ$^*$R-tree}, and the notion of \emph{spirality} for orthogonal components. We recall below these two concepts; see \Cref{fi:spqr-spirality} for an illustration.

\begin{figure}[t]
	\centering
    \subfigure[]{\includegraphics[height=0.4\columnwidth,page=1]{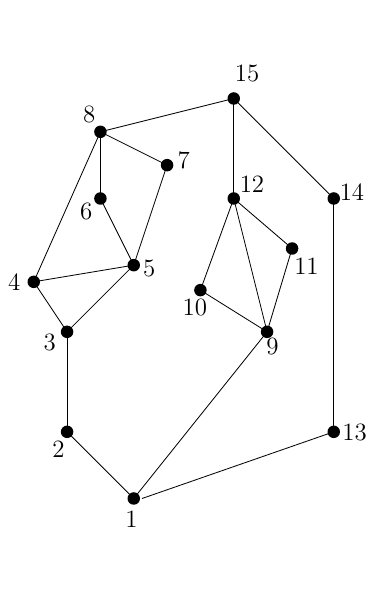}\label{fi:spqr-spirality-a}}
    \hfil
    \subfigure[]{\includegraphics[height=0.4\columnwidth,page=2]{spqr-spirality}\label{fi:spqr-spirality-b}}
	\caption{(a) A biconnected planar graph $G$. (b) The SPQ$^*$R-tree of $G$ rooted at the Q$^*$-node of the chain (1,13,14,15). The skeletons of three nodes are shown; the virtual edges are dashed and the reference edge is thicker.}
 \label{fi:spqr-spirality}
\end{figure}

\smallskip\noindent{\bf SPQ$^*$R-trees.}
Assume that $G$ is not a simple cycle (otherwise computing a bend-minimum orthogonal representation of $G$ is trivial). In an SPQ$^*$R-tree of $G$ we do not distinguish between Q-nodes and series of edges, and represent both of them with a type of node called \emph{Q$^*$-node}. More precisely, each degree-1 node of $T$ is a Q$^*$-node that corresponds to a maximal chain of edges of $G$ (possibly a single edge) starting and ending at vertices of degree larger than two and passing through a sequence of degree-2 vertices only (possibly none). If~$\nu$ is an S-, a P-, or an R-node, and if $\mu$ is a Q$^*$-node child of $\nu$, the edge of $\skel(\nu)$ corresponding to $\mu$ is virtual if $\mu$ corresponds to a chain of at least two edges, else it is a real edge.
If $T$ is rooted at a certain Q$^*$-node $\rho$, the chain corresponding to $\rho$ is the \emph{reference chain of $G$}. The definitions of pertinent graphs (or components) for the nodes of a rooted SPQ$^*$R-tree extend naturally from those of a rooted SPQR-tree. In particular $G_\rho$ is the pertinent graph of $\rho$, i.e., the reference chain of $G$.
Moreover, as in~\cite{DBLP:conf/compgeom/ChaplickGFGRS22,DBLP:journals/siamcomp/BattistaLV98,DBLP:journals/siamdm/DidimoGL09,DBLP:conf/gd/DidimoKLO22,Didimo2023,DBLP:conf/isaac/DidimoL98,DBLP:journals/jcss/DidimoLP19}, we assume to work with a \emph{normalized} SPQ$^*$R-tree, in which every S-node has exactly two children. Note that every SPQ$^*$R-tree can be easily transformed into a normalized SPQ$^*$R-tree by recursively splitting an S-node with more than two children into multiple S-nodes with two children. Observe that in this case an S-node may have an S-node as a child. If $G$ has $n$ vertices, a normalized SPQ$^*$R-tree of $G$ still has $O(n)$ nodes and can be easily computed in $O(n)$ time from an SPQ$^*$R-tree of $G$.  

\smallskip\noindent{\bf Orthogonal Spirality.} Let $H$ be an orthogonal representation of~$G$.
If $P^{uv}$ is a path from a vertex $u$ to a vertex $v$ in $H$, the \emph{turn number of $P^{uv}$}, denoted by $n(P^{uv})$, is the number of right turns minus the number of left turns along $P^{uv}$, while moving from $u$ to $v$. Note that a turn can occur at either a bend or a vertex.
Let $T$ be a rooted normalized SPQ$^*$R-tree of~$G$ and let $\nu$ be a node of~$T$. Let $H_\nu$ be the restriction of $H$ to $G_\nu$, and denote by $\{u,v\}$  the poles of $\nu$, conventionally ordered according to an $st$-numbering of $G$, where $s$ and $t$ are the poles of the root of $T$.
For each $w \in \{u,v\}$, let  $\indeg_\nu(w)$ and $\outdeg_\nu(w)$ be the degree of $w$ inside and outside $H_\nu$, respectively. Define two, possibly coincident, \emph{alias vertices} of $w$, denoted by $w'$ and $w''$, as follows:
$(i)$ if $\indeg_\nu(w)=1$, then $w'=w''=w$;
$(ii)$ if $\indeg_\nu(w)=\outdeg_\nu(w)=2$, then $w'$ and $w''$ are dummy vertices, each splitting one of the two distinct edge segments incident to~$w$ outside~$H_{\nu}$;
$(iii)$ if $\indeg_\nu(w)>1$ and $\outdeg_\nu(w)=1$, then $w'=w''$ is a dummy vertex that splits the edge segment incident to $w$ outside~$H_{\nu}$.

Let $A^w$ be the set of distinct alias vertices of a pole $w$. Let $P^{uv}$ be any simple path from~$u$ to~$v$ in $H_\nu$, and let~$u'$ and~$v'$ be alias vertices of~$u$ and~$v$, respectively. The path $S^{u'v'}$ obtained by concatenating $(u',u)$, $P^{uv}$, and $(v,v')$ is a \emph{spine} of $H_\nu$. 
The \emph{spirality} $\sigma(H_\nu)$ of $H_\nu$, introduced in~\cite{DBLP:journals/siamcomp/BattistaLV98}, is either an integer or a semi-integer number, defined based on the following cases:
%
$(i)$ if $A^u=\{u'\}$ and $A^v=\{v'\}$ then $\sigma(H_\nu) = n(S^{u'v'})$;
$(ii)$ if $A^u=\{u'\}$ and $A^v=\{v',v''\}$ then $\sigma(H_\nu) = \frac{1}{2} \big (n(S^{u'v'}) + n(S^{u'v''}) \big )$;
$(iii)$ if $A^u=\{u',u''\}$ and $A^v=\{v'\}$ then $\sigma(H_\nu) = \frac{1}{2} \big ( n(S^{u'v'}) + n(S^{u''v'})\big )$;
$(iv)$ if $A^u=\{u',u''\}$ and $A^v=\{v',v''\}$, assume w.l.o.g. that $(u,u')$ is the edge before $(u,u'')$ counterclockwise around $u$ and that $(v,v')$ is the edge before $(v,v'')$ clockwise around $v$; then $\sigma(H_\nu) = \frac{1}{2} \big (n(S^{u'v'}) + n(S^{u''v''})\big )$.

\Cref{fi:spirexaple} shows the graph $G$ of \Cref{fi:spqr-spirality-a}, two P-components $G_{\mu_1}$ and $G_{\mu_2}$ highlighted, and an orthogonal representation $H$ of $G$. The orthogonal components $H_{\mu_1}$ and $H_{\mu_2}$ in $H$ have respectively spiralities $\sigma_{\mu_1}=-1$, see Case~(iv), and  $\sigma_{\mu_2}=0$, see Case~(i). \Cref{fi:ftp-root-a} shows an orthogonal S-component with spirality $\frac{3}{2}$, see Case~(ii). 

\begin{figure}[t]
	\centering
  \includegraphics[height=0.42\columnwidth,page=3]{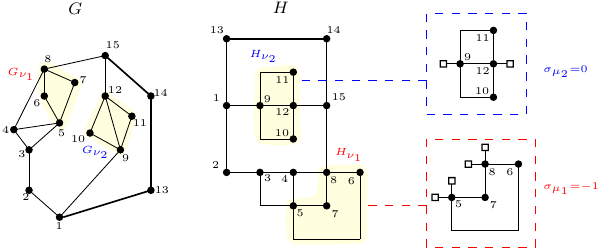}  
	\caption{Illustration of the concept of spirality.}
 \label{fi:spirexaple}
\end{figure}

\smallskip
Note that, the spirality of $H_\nu$ does not vary if we choose a different path~$P^{uv}$. Also, it is proved in~\cite{DBLP:journals/siamcomp/BattistaLV98} that a component $H_\nu$ of~$H$ can always be substituted by any other representation $H'_\nu$ of $G_\nu$ with the same spirality, getting a new valid orthogonal representation with the same set of bends on the edges of $H$ that are not in $H_\nu$ (see~\cite{DBLP:journals/siamcomp/BattistaLV98} and also Theorem~1 in \cite{Didimo2023}). For brevity, in the remainder of the paper we often denote by $\sigma_\nu$ the spirality of a component $H_\nu$.

\medskip\noindent{\bf Testing algorithm.}
Let $T$ be a rooted normalized SPQ$^*$R-tree of the input graph~$G$, and let $b$ be a non-negative integer. 
An orthogonal representation of~$G$ (resp. of a component $G_\nu$ of~$G$) with at most~$b$ bends is a \emph{$b$-orthogonal representation} of~$G$ (resp. of~$G_\nu$). 
\cref{th:biconnected} will prove that testing if $G$ admits a $b$-orthogonal representation is FPT parametrized by $k+b$. We start with a key property.

\begin{lemma}
\label{le:spirbounded}
Let $\nu$ be a node of $T$ and $H$ be a $b$-orthogonal representation of $G$. The spirality $\sigma_\nu$ of the restriction $H_\nu$ of $H$ to $G_\nu$ belongs to $[-k-b-2, k+b+2]$.
\end{lemma}
\begin{proof}
We consider the case $\sigma_\nu>0$ and prove that $\sigma_\nu \le k+b+2$. When $\sigma_\nu$ is negative, the proof that $-k-b-2 \le \sigma_\nu$ is symmetric. If $ 0 \leq \sigma_\nu \le 2$, the lemma is trivially true. Assume now that $\sigma_\nu > 2$. Denote by ${u,v}$ the two poles of $\nu$. The boundary of $H_\nu$ can be split into two paths from $u$ to $v$, which we call the left path $P^l$ and the right path $P^r$ of $H_\nu$. Namely, $P^l$ (resp. $P^r$) is the path from $u$ to $v$ while walking clockwise (resp. counterclockwise) on the external boundary of $H_\nu$. Note that, these paths may share some edges.
By the definition of spirality, we have that $\sigma_\nu \leq n(P^l) + 2$, where $n(P^l)$ is the turn number of $P^l$. Since we are assuming $\sigma_\nu > 2$, this implies $n(P^l)>0$. Also, each right turn of $P^l$ is a $270^o$ angle in the external face of $H_\nu$, hence it is either a degree-2 vertex or a bend. Hence, $n(P^l) \le k+b$ and, consequently, $\sigma_\nu \le k+b+2$. 
See for example 
\cref{fi:ftp-root-a}, which depicts an orthogonal representation $H_\nu$ of $G_\nu$ with spirality $\sigma_\nu = \frac{3}{2}$ and $n(P^l)=3$. Path $P^l$ has one degree-2 vertex and three bends.  
\end{proof}

\smallskip\noindent{\bf Spirality sets.} For each node $\nu$ of $T$, we define a function $X_\nu$ that maps each target spirality value $\sigma_\nu$ for a representation of $G_\nu$ to a pair $(b_\nu,H_\nu)$ such that $b_\nu$ is the minimum number of bends over all orthogonal representations of $G_\nu$ with spirality $\sigma_\nu$, and $H_\nu$ is a $b_\nu$-orthogonal representation of $G_\nu$ with spirality $\sigma_\nu$. We say that $b_\nu$ is the \emph{cost} of $\sigma_\nu$ for $\nu$ and that $H_\nu$ is the \emph{representative} for $\nu$ given $\sigma_\nu$. 
If the minimum number of bends of an orthogonal representation of $G_\nu$ having spirality $\sigma_\nu$ is higher than a target maximum number of bends $b$, we write $X_\nu(\sigma_\nu)=(\infty,\emptyset)$. 
The set $\Sigma_\nu = \{ (\sigma_\nu, X_\nu(\sigma_\nu)) \;|\; X_\nu(\sigma_\nu) \neq (\infty,\emptyset)\}$ is called 
the \emph{spirality set of $\nu$}. 

If for some node $\nu$ the spirality set $\Sigma_\nu$ is empty, 
the instance can be safely rejected. Also, by \cref{le:spirbounded}, for each node $\nu$ of $T$,  we can assume $|\sigma_\nu|\le k+b+2$. Hence, we have $|\Sigma_\nu|=O(k+b)$.
%
In~\cite{DBLP:journals/siamcomp/BattistaLV98} it is shown that $\Sigma_\nu$  can be computed in time: $(i)$ $O(|\Sigma_\nu|)$ if $\nu$ is a Q$^*$-node; $(ii)$ $O(|\Sigma_\nu|)$ if $\nu$ is a P-node and $O(|\Sigma_\nu|^2)$ if $\nu$ is an S-node, and if we know the spirality sets of the children of $\nu$. For each pair $(\sigma_\nu,X_\nu(\sigma_\nu))$ in $\Sigma_\nu$, the representative $H_\nu$ for $\nu$ is encoded in constant time by suitably linking the representatives for the children of $\nu$ and adding a constant-size information that specifies the values of the angles at the poles~of~$\nu$.   
By \Cref{le:spirbounded} we can restate the results in~\cite{DBLP:journals/siamcomp/BattistaLV98} as follows.   


\begin{lemma}
\label{le:Q-computation}
If $\nu$ is a Q$^*$-node of $T$ then $\Sigma_\nu$ can be computed in $O(k+b)$ time.
\end{lemma}



\begin{lemma}
\label{le:P-computation}
Let $\nu$ be a P-node of $T$ with children $\mu_1,...,\mu_h$ ($h \in \{2,3\}$). Given $\Sigma_{\mu_i}$ for each $i \in \{1, \dots, h\}$, there is an $O(k+b)$-time algorithm that computes~$\Sigma_\nu$.
\end{lemma}

\begin{lemma}
\label{le:S-computation}
Let $\nu$ be an S-node of $T$ with children $\mu_1$ and $\mu_2$. Given $\Sigma_{\mu_1}$ and $\Sigma_{\mu_2}$, there is an $O((k+b)^2)$-time algorithm that computes~$\Sigma_\nu$.
\end{lemma}


Let $\rho$ be the root of $T$ and let $H$ be an orthogonal representation of $G$ where $G_\rho$ (i.e., the reference chain) is on the external face. Without loss of generality, we assume that $G_\rho$ is to the right of $H$, i.e., for each simple cycle containing the end-vertices
$u$ and $v$ of $G_\rho$, visiting this cycle clockwise we encounter $u$, $v$, and $G_\rho$ in this order. 
In any orthogonal representation $H$ of $G$, the number of right turns minus the number of left turns encountered traversing clockwise any cycle of $G$ in $H$ equals four~\cite{DBLP:books/ph/BattistaETT99}.
Hence, the following lemma immediately holds.

\begin{lemma}
\label{le:chain}
Let $H$ be an orthogonal representation of $G$, let $T$ be the SPQ$^*$R-tree rooted at $\rho$, and let $\nu$ be the root child of $T$. Denoted by $n(H_\rho^{uv})$ the turn number of the chain $G_\rho$ in $H$, moving from $u$ to $v$, we have
$\sigma(H_\nu)-n(H_\rho^{uv})=4$. 
\end{lemma}

\noindent\cref{fi:ftp-root-b} illustrates \cref{le:chain}, where $n(H_\rho^{uv})=-5$ and $\sigma(H_\nu)=-1$. The next lemma gives the time complexity needed to compute $\Sigma_\nu$ when $\nu$ is an R-node. 

\begin{figure}[t]
	\centering
    \subfigure[]{\includegraphics[width=0.48\columnwidth,page=1]{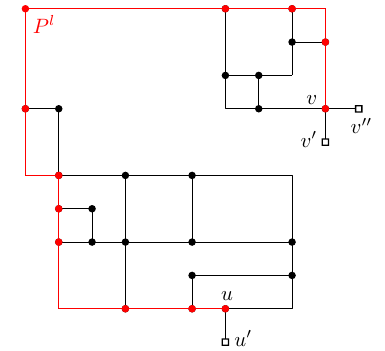}\label{fi:ftp-root-a}}
    \subfigure[]{\includegraphics[width=0.48\columnwidth,page=2]{ftp-root}\label{fi:ftp-root-b}}
	\caption{(a) Illustration for the proof of \cref{le:spirbounded}. (b) Illustration for \cref{le:chain}.}
 \label{fi:root}
\end{figure} 

\begin{restatable}{lemma}{leRcomputation}
\label{le:R-computation}
Let $\nu$ be an R-node of $T$ and let $\mu_1, \dots ,\mu_d$ be its children that are not $Q^*$-nodes. Given $\Sigma_{\mu_i}$ for each $i \in \{1, \dots, d\}$, there exists an algorithm that computes $\Sigma_\nu$ in $O(2^{p \cdot \log p}) \cdot n^{O(1)}$ time, with $p=k+b$.
\end{restatable}
\begin{proof}[Sketch]
We prove the lemma in the case where $G_\nu$ has two alias vertices. The other cases can be proved with similar arguments, see the appendix. For every target value of spirality $\sigma_\nu \in [-k-b-2,k+b+2]$, we compute $X_\nu(\sigma_\nu)$ as follows.
For each of the two possible embeddings of $\skel(\nu)$, consider the component $G_\nu$ enhanced with a dummy edge $e$, connecting the poles $u$ and $v$ of $\nu$, in such a way that $e$ is on the external face of $G_\nu$ and to the right of $G_\nu$. Let $G'$ be the resulting graph. A \emph{valid tuple} is a tuple $(\sigma_1,...,\sigma_d)$ such that $\sum_{i=1}^db_i\le b$, where $X_{\mu_i}(\sigma_i)=(b_i,H_i)$ for every $i\in \{1,\dots, d\}$. 
For each valid tuple, we perform the following procedure. Let $H_e$ be an orthogonal representation of $e$ such that $n(H_e^{uv})=\sigma_\nu-4$, where $H_e^{uv}$ is the turn number of $H_e$ going from $u$ to $v$. This means that $H_e^{uv}$ has $\sigma_\nu-4$ bends, each turning to the left if $\sigma_\nu-4<0$, or each turning to the right if $\sigma_\nu-4>0$. Also, let $J=e \cup G_{\mu_1} \cup \dots \cup G_{\mu_d}$ and let $H_J$ be an orthogonal representation of the graph $J$ such that: $(i)$ the restriction of $H_J$ to $G_{\mu_i}$ coincides with $H_{\mu_i}$ for each $i\in \{1,\dots, d\}$; and $(ii)$ the restriction of $H_J$ to $e$ coincides with $H_e$. We compute an $H_J$-constrained orthogonal representation $H'$ of $G'$ with the minimum number of bends in polynomial time by \cref{le:planebendmin}, if it exists.  Consider an SPQ$^*$R-tree $T'$ of $G'$ rooted at the Q$^*$-component $\xi$ representing $e$. 
Note that $\nu$ is the root child of $T'$.
Hence, by \cref{le:chain}, if $H'$ exists then its restriction $H_\nu$ to $G_\nu$ has spirality $\sigma_\nu$. By \cref{le:planebendmin}, $H_\nu$ is a bend-minimum orthogonal representation of $G_\nu$ having spirality $\sigma_\nu$ and the spirality of the restriction of $H_\nu$ to $G_{\mu_i}$ is $\sigma_i$, for each $i\in \{1,\dots, d\}$. 
Denote by $b_\nu$ the number of bends of $H_\nu$.
If $b_\nu > b$ for each of the two planar embeddings of $\skel(\nu)$, then $X_\nu(\sigma_\nu)=(\infty,\emptyset)$ and we do not insert $(\sigma_\nu,X_\nu(\sigma_\nu))$ in the spirality set $\Sigma_\nu$. Else, for the two embeddings of $\skel(\nu)$ we retain the representation $H_\nu$ of minimum cost $b_\nu$, set $X_\nu(\sigma_\nu)=(b_\nu,H_\nu)$, and insert $(\sigma_\nu,X_\nu(\sigma_\nu))$ in $\Sigma_\nu$. 

\smallskip \noindent\textsf{Correctness.} The correctness of the procedure above follows by these facts: $(i)$ to construct a bend-minimum representation $H_\nu$ with spirality $\sigma_\nu$, we consider all possible combinations of values of spiralities for $G_{\mu_1}, \dots, G_{\mu_d}$; $(ii)$ thanks to the interchangeability of orthogonal components with the same spirality, for each such combination we aim to construct $H_\nu$ so that it contains a minimum-bend representation of each $G_{\mu_i}$ with its target value of spirality; $(iii)$ if $b_\nu \leq b$, the spirality determined by $H_\nu$ for each child $\mu$ of $\nu$ that corresponds to a Q$^*$-component is surely in the spirality set of $\mu$; $(iv)$ By \cref{le:spirbounded} it suffices to test the existence of an $H_\nu$ for each target value of spirality $\sigma_\nu\in [-k-b-2,k+b+2]$.

\smallskip \noindent\textsf{Time-complexity.} By \cref{le:planebendmin}, for each tuple the time required for the computation is $O(n^\frac{7}{4}\sqrt{\log n})$. By \cref{le:spirbounded}, we consider $O(d^{k+b})$ spirality values.
Moreover, we can assume that $d\le k+b$, otherwise the instance can be rejected. More precisely, since $G_{\mu_i}$ ($i\in \{1,\dots, d\}$) contains at least one cycle, any orthogonal representation of $G_{\mu_i}$ requires four $270^\circ$ (i.e., reflex) angles on the external face, and at most two of these angles can occur at the poles of $G_{\mu_i}$. Since each reflex angle that does not occur at a pole of $G_{\mu_i}$ requires either a degree-2 vertex or a bend, we necessarily have $d \leq k+b$ if $H_\nu$ exists. 
Hence, there are $O(d^{k+b})=O((k+b)^{k+b})=O(2^{(k+b) \cdot \log(k+b)})$ valid tuples.
\end{proof}

We can now prove the main result of this section.

\begin{theorem}
\label{th:biconnected}
Let $G$ be a biconnected graph with $k$ degree-2 vertices, let $b$ be an integer, and let $p=k+b$. There exists an $O(2^{p \cdot \log p}) \cdot n^{O(1)}$-time algorithm 
that tests whether $G$ admits a $b$-orthogonal representation, and that computes one with the minimum number of bends in the positive case.
\end{theorem}
\begin{proof}
If $G$ is a simple cycle the test is trivial. Otherwise, let $T$ be the SPQ$^*$R-tree of $G$. For each Q$^*$-node $\rho$ of $T$, root $T$ at $\rho$. Recall that this rooted tree describes all planar embeddings of $G$ with the reference chain $G_\rho$ on the external face. We perform a post-order visit of~$T$. We first compute $\Sigma_\nu$ for each leaf $\nu$ of $T$, that is, for each Q$^*$-node distinct from $\rho$. We can do this in $O(n(k+b))$ time by \cref{le:Q-computation}. During this visit of~$T$, for every internal non-root node~$\nu$ of~$T$ the algorithm computes~$\Sigma_\nu$ by using the spirality sets of the children of~$\nu$, exploiting \cref{le:P-computation,le:S-computation,le:R-computation}, depending on whether $\nu$ is a P-, an S-, or an R-node, respectively. If the spirality set $\Sigma_\nu$ of $\nu$ is empty, then $G$ does not have a $b$-orthogonal representation with the given reference edge $G_\rho$ on the external face. In this case the algorithm stops visiting $T$ rooted at $\rho$, and starts visiting $T$ rooted at another Q$^*$-node. 
Suppose that the algorithm reaches the root child $\nu$ of $T$, computes the spirality set $\Sigma_\nu$, and this spirality set is not empty. 
Denote by $u$ and $v$ the poles of $\nu$.
Then, the algorithm considers each pair $(\sigma_\nu, X_\nu(\sigma_\nu)=(b_\nu,H_\nu))$ in the spirality set of $\nu$. \cref{le:chain} implies that $G$ admits an orthogonal representation $H$ with $G_\rho$ on the external face whose restriction to $G_\nu$ in $H$ has spirality $\sigma_\nu$ if and only if the restriction to $G_\rho$ in $H$ has an orthogonal representation $H_\rho$ such that $n(H_\rho^{uv})=\sigma_\nu-4$. Let $n_\rho$ be the number of vertices of $G_\rho$ different from $u$ and $v$. Since $G_\rho$ is a chain of edges, we have that such $H_\rho$ has $b_\rho=0$ bends if $n_\rho\ge |\sigma_\nu-4|$; otherwise it has $b_\rho=|\sigma_\nu-4|-n_\rho$ bends and we just check if $b_\nu+b_\rho \le b$. By the reasoning above, $G$ admits a $b$-orthogonal representation $H$ if and only if such a test is positive for some $\sigma_\nu$. Also, if this is true, we can construct $H$ by attaching $H_\nu$ to $H_\rho$, and since for each node $\nu$ we have used the representation $H_\nu$ with minimum number of bends among those with the same spirality, $H$ has the minimum number of bends over all $b$-orthogonal representations of $G$.

By \cref{le:P-computation,le:S-computation,le:R-computation} for each node we spend  $O(2^{p \cdot \log p}) \cdot n^{O(1)}$ time, and the condition of \Cref{le:chain} at the root level can be tested in $O(n)$ time. Since $T$ has $O(n)$ nodes and  $G$ has $O(n)$ rooted SPQ$^*$R-trees, the statement follows.\end{proof}


\section{The General Case}
\label{se:generalcase}
We assume first that $G$ is connected but not biconnected; we then consider non-connected graphs in the proof of \cref{th:generalcase}, the proof of which is in the appendix. A biconnected component of $G$ is also called a \emph{block}. There are two main difficulties in extending the result of \Cref{th:biconnected} when $G$ is not biconnected. In terms of drawability, some angle constraints may be required at the cutvertices of the input graph $G$. Namely, one cannot simply test the representability of each single block independently, as it might be impossible to merge the orthogonal representations of the different blocks into an orthogonal representation of $G$ without additional angle constraints at the cutvertices. In terms of vertex-degree, $G$ might also have degree-1 vertices and, more importantly, a vertex of degree three or four in $G$ can be a degree-2 vertex in a block; hence, the sum of the degree-2 vertices in all blocks can be larger than the number of degree-2 vertices in $G$. To manage angle constraints at the cutvertices, we adopt a variant of a strategy used in a previous work on series-parallel graphs~\cite{DBLP:conf/gd/DidimoKLO22}, and extend it to also consider rigid components. About the number of degree-2 vertices, similarly to~\cite{DBLP:journals/jcss/GiacomoLM22}, we extend the parameter $k$ to include both vertices of degree one and two, and we observe that each block of $G$ must have $O(k+b)$ degree-2 vertices if $G$ admits a $b$-orthogonal representation.

\smallskip\noindent{\bf Block-cutvertex tree.} Let $G$ be a connected graph
having $k$ vertices of degree at most two, and let $b$ be a non-negative integer. To deal with the possible planar embeddings of $G$, we use the \emph{BC-tree} (\emph{block-cutvertex tree}) of $G$, and combine it with the SPQ$^*$R-trees of its blocks. The BC-tree $\cal T$ of $G$ describes the decomposition of~$G$ in terms of its blocks and cutvertices; refer to~\Cref{fi:bc-tree}. Each node of $\cal T$ represents either a block or a cutvertex of $G$. A \emph{block-node} (resp. a \emph{cutvertex-node}) of $\cal T$ represents a block (resp. a cutvertex) of $G$. There is an edge between two nodes of $\cal T$ if one of them represents a cutvertex and the other represents a block that contains the cutvertex.
If $B_1, \dots, B_q$ are the blocks of $G$ $(q \geq 2)$, we denote by $\beta(B_i)$ the block-node of $\cal T$ corresponding to $B_i$ $(1 \leq i \leq q)$ and by ${\cal T}_{B_i}$ the tree $\cal T$ rooted at $\beta(B_i)$. For a cutvertex $c$ of $G$, we denote by $\chi(c)$ the node of $\cal T$ that corresponds to~$c$. 
Each ${\cal T}_{B_i}$ describes a class of planar embeddings of $G$ such that, for each non-root node $\beta(B_j)$ $(1 \leq j \leq q)$ with parent node $\chi(c)$ and grandparent node $\beta(B_k)$, the cutvertex $c$ and $B_k$ lie on the external face of $B_j$. 
An \emph{orthogonal representation of $G$ with respect to ${\cal T}_{B_i}$} is an orthogonal representation whose planar embedding belongs to the class described by~${\cal T}_{B_i}$. Therefore, to test whether $G$ admits a $b$-orthogonal representation we have to test if $G$ admits a $b$-orthogonal representation with respect to ${\cal T}_{B_i}$ for some $i \in \{1, \dots, q\}$. 

\smallskip\noindent{\bf Testing algorithm.} For any fixed $i \in \{1, \dots, q\}$, testing whether $G$ has a $b$-orthogonal representation with respect to ${\cal T}_{B_i}$, requires to compute a bend-minimum representation for each block $B_j$ subject to angle constraints at the cutvertices of $B_j$, so that we can merge the representation of $B_j$ with those of its adjacent blocks. For example, consider the orthogonal representation $H$ in \Cref{fi:bc-tree-a}, whose embedding is in the class described by the rooted BC-tree ${\cal T}_{B_4}$ of \Cref{fi:bc-tree-b}. $H$ has blocks $B_1, \dots B_5$ (where $B_3$ and $B_5$ are trivial blocks) and cutvertices $c$, $c'$, and $c''$. Note that the representation of $B_1$ must have $c$ on its external face with a flat angle, so to accommodate the representation of $B_3$. Also, $c'$ must have a reflex angle in the representation of $B_1$, so to accommodate the representation of $B_2$. Conversely, the two orthogonal representations in \Cref{fi:bc-tree-c} have embeddings that do not adhere to any rooted version of $\cal T$, and they cannot be merged, because neither $c$ nor $c''$ lie on the external face of their components.  A detailed description of all types of angle constraints that may be needed at the cutvertices of a block is given in the appendix.

\begin{figure}[t]
    \centering
    \subfigure[]{\includegraphics[height=0.33\columnwidth,page=1]{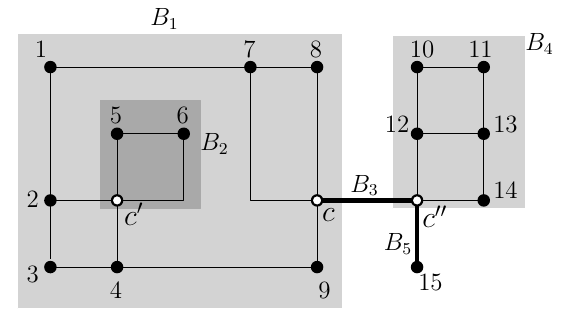}\label{fi:bc-tree-a}}
    \hfil
    \subfigure[]{\includegraphics[height=0.33\columnwidth,page=2]{bc-tree}\label{fi:bc-tree-b}}
    \hfil
    \subfigure[]{\includegraphics[height=0.36\columnwidth,page=3]{bc-tree}\label{fi:bc-tree-c}}
    
    \caption{(a) An orthogonal representation $H$ of a 1-connected graph and its blocks. (b) The BC-tree ${\cal T}_{B_4}$; the embedding of $H$ is in the class described by ${\cal T}_{B_4}$. (c) Orthogonal representations of $B_1$ and $B_4$ that cannot be glued planarly.}
 \label{fi:bc-tree}
\end{figure} 


\medskip
A \emph{$b$-constrained orthogonal representation of $B_j$} is a representation   
$H_{B_j}$ with at most $b$ bends and that fulfills the required angle constraints at the cutvertices of $G$ that belong to $B_j$. We prove the following (see appendix for details).

\begin{restatable}{lemma}{leconstrainedblock}
	\label{le:constrained-block}
	Let $G$ be a graph with $k$ vertices of degree at most two and let $b$ be a non-negative integer. Let $B_1, \dots, B_q$ be the blocks of $G$, let ${\cal T}_{B_i}$ be the BC-tree of $G$ rooted at $\beta(B_i)$ (for some $i \in \{1, \dots, q\}$), and let $B_j$ be any block of $G$ (possibly $i=j$).
	There exists an $O(2^{p \cdot \log p}) \cdot n^{O(1)}$-time  
	algorithm, with $p=k+b$, that tests whether $B_j$ admits a $b$-constrained orthogonal representation, and that computes one with minimum number of bends in the positive case.
\end{restatable}
\begin{proof}[Sketch]
    We first show that the number of degree-2 vertices of $B_j$ is at most $2k + b$ if $G$ admits a $b$-orthogonal representation with respect to ${\cal T}_{B_i}$. Then we show how to apply the algorithm of \Cref{th:biconnected} on the rooted SPQ$^*$R-trees of $B_j$, handling the constraints at the cutvertices.    
\end{proof}

\noindent For all possible rooted BC-trees ${\cal T}_{B_i}$, the testing algorithm checks if $G$ has a $b$-orthogonal representation with respect to ${\cal T}_{B_i}$: it uses \Cref{le:constrained-block} on every block $B_j$ and finally checks that the sum of the bends is at most $b$ (see appendix).

\begin{restatable}{theorem}{thgeneralcase}
\label{th:generalcase}
Let $G$ be a graph with $k$ vertices of degree at most 2, let $b$ be an integer, and let $p=k+b$. There exists an $O(2^{p \cdot \log p}) \cdot n^{O(1)}$-time algorithm 
that tests whether $G$ admits a $b$-orthogonal representation, and that computes one with the minimum number of bends in the positive case.
\end{restatable}


\begin{corollary}
Let $G$ be a graph with $k$ vertices of degree at most two. There exists an $O(2^{k \cdot \log k}) \cdot n^{O(1)}$-time algorithm 
that tests whether $G$ admits a rectilinear representation, and that computes one in the positive case.
\end{corollary}

\bibliographystyle{splncs04}
\bibliography{bibliography}

\clearpage
\appendix

\section{Additional Material for \Cref{se:bico}}\label{se:app-bico}

\leRcomputation*
\begin{proof}
For every target value of spirality $\sigma_\nu \in [-k-b-2,k+b+2]$, we compute $X_\nu(\sigma_\nu)$ as follows. A \emph{valid tuple} is a tuple $(\sigma_1,...,\sigma_d)$ such that $\sum_{i=1}^db_i\le b$, where $X_{\mu_i}(\sigma_i)=(b_i,H_i)$ for every $i\in \{1,\dots, d\}$. Let $u$ and $v$ be the poles of $\nu$. For each of the two possible embeddings of $\skel(\nu)$, consider the component $G_\nu$ enhanced with a dummy edge $e$, connecting the poles $u$ and $v$ of $\nu$, in such a way that $e$ is on the external face of $G_\nu$ and to the right of $G_\nu$. Let $G'$ be the resulting graph.
We have four cases depending on the number of alias vertices of $u$ and $v$: $(a)$~$|A^u|=|A^v|=1$; $(b)$~$|A^u|=1$ and $|A^v|=2$; $(c)$~$|A^u|=2$ and $|A^v|=1$;  $(d)$~$|A^u|=|A^v|=2$.

\begin{itemize}
\item \texttt{Case (a).} For each valid tuple, we perform the following procedure. Let $H_e$ be an orthogonal representation of $e$ such that $n(H_e^{uv})=\sigma_\nu-4$, where $H_e^{uv}$ is the turn number of $H_e$ going from $u$ to $v$. This means that $H_e^{uv}$ has $\sigma_\nu-4$ bends, each turning to the left if $\sigma_\nu-4<0$, or each turning to the right if $\sigma_\nu-4>0$. Also, let $J=e \cup G_{\mu_1}\cup...\cup G_{\mu_d}$ and let $H_J$ be an orthogonal representation of the graph $J$ such that: $(i)$ the restriction of $H_J$ to $G_{\mu_i}$ coincides with $H_{\mu_i}$ for each $i\in \{1,\dots, d\}$; and $(ii)$ the restriction of $H_J$ to $e$ coincides with $H_e$. We compute an $H_J$-constrained orthogonal representation $H'$ of $G'$ with the minimum number of bends in polynomial time by \cref{le:planebendmin}, if it exists.  Consider an SPQ$^*$R-tree $T'$ of $G'$ rooted at the Q$^*$-component $\xi$ representing $e$. 
Note that $\nu$ is the root child of $T'$.
Hence, by \cref{le:chain}, if $H'$ exists then its restriction $H_\nu$ to $G_\nu$ has spirality $\sigma_\nu$. By \cref{le:planebendmin}, $H_\nu$ is a bend-minimum orthogonal representation of $G_\nu$ having spirality $\sigma_\nu$ and the spirality of the restriction of $H_\nu$ to $G_{\mu_i}$ is $\sigma_i$, for each $i\in \{1,\dots, d\}$. 

\item \texttt{Case (b).} We add vertex $v'$ and edge $e_v=(v,v')$ to $G'$ so that $e$ is the edge after $e_v$ in the clockwise order of the edges incident to $v$. We proceed similarly as we did in Case~$(a)$. For each valid tuple, we perform the following procedure. Let $H_e$ be an orthogonal representation of $e$ such that $n(H_e^{uv})=\sigma_\nu-\frac{7}{2}$, where $H_e^{uv}$ is the turn number of $H_e$ going from $u$ to $v$. Observe that, since $\sigma_\nu$ is a semi-integer number, $n(H_e^{uv})$ is an integer number. Let $J=e \cup G_{\mu_1}\cup...\cup G_{\mu_d}$ and let $H_J$ be an orthogonal representation of the graph $J$ such that: $(i)$ the restriction of $H_J$ to $G_{\mu_i}$ coincides with $H_{\mu_i}$ for each $i\in \{1,\dots, d\}$; and $(ii)$ the restriction of $H_J$ to $e$ coincides with $H_e$. We compute an $H_J$-constrained orthogonal representation $H'$ of $G'$ with the minimum number of bends in polynomial time by \cref{le:planebendmin}, if it exists.  Consider an SPQ$^*$R-tree $T'$ of $G'$ rooted at the Q$^*$-component $\xi$ representing $e$.  Suppose that $H'$ exists. We can assume that $e_v$ has 0 bends in $H'$, since it is an edge incident to a degree-1 vertex. Let $H''$ be the orthogonal representation obtained from $H'$ removing $e_v$. By \cref{le:chain}, the restriction of $H''$ to $G_\nu$ has spirality $\sigma_\nu+\frac{1}{2}$. By definition of spirality and by construction of $G'$, the restriction $H_\nu$ of $H'$ to $G_\nu$ has spirality $\sigma_\nu$. By \cref{le:planebendmin}, $H_\nu$ is a bend-minimum orthogonal representation of $G_\nu$ having spirality $\sigma_\nu$ and the spirality of the restriction of $H_\nu$ to $G_{\mu_i}$ is $\sigma_i$, for each $i\in \{1,\dots, d\}$. 

\item \texttt{Case (c).} Symmetric to Case (b).

\item \texttt{Case (d).} We add vertex $v'$ and edge $e_v$ to $G'$ as in Case~(b). Also, we add vertex $u'$ and edge $e_u=(u,u')$ to $G'$ so that $e$ is the edge before $e_u$ in the clockwise order of the edges incident to $u$. For each valid tuple, we perform the following procedure. Let $H_e$ be an orthogonal representation of $e$ such that $n(H_e^{uv})=\sigma_\nu-1$, where $H_e^{uv}$ is the turn number of $H_e$ going from $u$ to $v$. Let $J=e \cup G_{\mu_1}\cup...\cup G_{\mu_d}$ and let $H_J$ be an orthogonal representation of the graph $J$ such that: $(i)$ the restriction of $H_J$ to $G_{\mu_i}$ coincides with $H_{\mu_i}$ for each $i\in \{1,\dots, d\}$; and $(ii)$ the restriction of $H_J$ to $e$ coincides with $H_e$. We compute an $H_J$-constrained orthogonal representation $H'$ of $G'$ with the minimum number of bends in polynomial time by \cref{le:planebendmin}, if it exists.  Consider an SPQ$^*$R-tree $T'$ of $G'$ rooted at the Q$^*$-component $\xi$ representing $e$.  Suppose that $H'$ exists. We can assume that $e_v$ and $e_u$ have 0 bends in $H'$. Let $H''$ be the orthogonal representation obtained from $H'$ removing $e_v$ and $e_u$. By \cref{le:chain}, the restriction of $H''$ to $G_\nu$ has spirality $\sigma_\nu+1$. By definition of spirality and by construction of $G'$, the restriction $H_\nu$ of $H'$ to $G_\nu$ has spirality $\sigma_\nu$. By \cref{le:planebendmin}, $H_\nu$ is a bend-minimum orthogonal representation of $G_\nu$ having spirality $\sigma_\nu$ and the spirality of the restriction of $H_\nu$ to $G_{\mu_i}$ is $\sigma_i$, for each $i\in \{1,\dots, d\}$. 
\end{itemize}

Denote by $b_\nu$ the number of bends of $H_\nu$. 
If $b_\nu > b$ for each of the two planar embeddings of $\skel(\nu)$, then $X_\nu(\sigma_\nu)=(\infty,\emptyset)$ and we do not insert $(\sigma_\nu,X_\nu(\sigma_\nu))$ in the spirality set $\Sigma_\nu$. Else, for the two embeddings of $\skel(\nu)$ we retain the representation $H_\nu$ of minimum cost $b_\nu$, set $X_\nu(\sigma_\nu)=(b_\nu,H_\nu)$, and insert $(\sigma_\nu,X_\nu(\sigma_\nu))$ in $\Sigma_\nu$. 

\smallskip \noindent\textsf{Correctness.} The correctness of the procedure above follows by these facts: $(i)$~to construct a bend-minimum representation $H_\nu$ with spirality $\sigma_\nu$, we consider all possible combinations of values of spiralities for $G_{\mu_1}, \dots, G_{\mu_d}$; $(ii)$~thanks to the interchangeability of orthogonal components with the same spirality, for each such combination we aim to construct $H_\nu$ in such a way that it contains a minimum-bend representation of each $G_{\mu_i}$ with its target value of spirality; $(iii)$~if $b_\nu \leq b$, the spirality determined by $H_\nu$ for each child $\mu$ of $\nu$ that corresponds to a Q$^*$-component is surely in the spirality set of $\mu$; $(iv)$~By \cref{le:spirbounded} it suffices to test the existence of an $H_\nu$ for each target value of spirality $\sigma_\nu\in [-k-b-2,k+b+2]$.

\smallskip \noindent\textsf{Time-complexity.} By \cref{le:planebendmin}, for each tuple the time required for the computation is $O(n^\frac{7}{4}\sqrt{\log n})$. By \cref{le:spirbounded}, we consider $O(d^{k+b})$ spirality values.
Moreover, we can assume that $d\le k+b$, otherwise the instance can be rejected. More precisely, since $G_{\mu_i}$ ($i\in \{1,\dots, d\}$) contains at least one cycle, any orthogonal representation of $G_{\mu_i}$ requires four $270^\circ$ (i.e., reflex) angles on the external face, and at most two of these angles can occur at the poles of $G_{\mu_i}$. Since each reflex angle that does not occur at a pole of $G_{\mu_i}$ requires either a degree-2 vertex or a bend, we necessarily have $d \leq k+b$ if $H_\nu$ exists. 
Hence, there are $O(d^{k+b})=O((k+b)^{k+b})=O(2^{(k+b) \cdot \log(k+b)})$ valid tuples.
\end{proof}

\section{Additional Material for \Cref{se:generalcase}}\label{se:app-general}

\noindent{\bf Types of angle constraints.} As mentioned in \Cref{se:generalcase}, for any fixed $i \in \{1, \dots, q\}$,  
testing whether $G$ has a $b$-orthogonal representation with respect to ${\cal T}_{B_i}$, requires to compute a bend-minimum representation for each block $B_j$ subject to angle constraints at the cutvertices of $B_j$, so that we can merge the representation of $B_j$ with those of the blocks adjacent to $B_j$.
In the following we describe the types of constraints that we must consider for the cutvertices of $G$ that belong to $B_j$ 
If $B_j \neq B_i$ ($\beta(B_j)$ is not the root of ${\cal T}_{B_i}$), we denote by $c$ the cutvertex of $B_j$ for which $\chi(c)$ is the parent of $\beta(B_j)$. Also, we always denote by $c'$ any cutvertex such that $\chi(c')$ is a child of~$\beta(B_j)$. Finally, let $\deg(c|B_j)$ and $\deg(c)$ (resp. $\deg(c'|B_j)$ and $\deg(c')$) be the degree of $c$ (resp. $c'$) in $B_j$ and the degree of $c$ (resp. $c'$) in $G$. 

\begin{itemize}
\item 
If $\deg(c'|B_j)=2$ and $\deg(c')=4$, there are two cases: $(i)$ $\deg(c'|B_k)=2$ for some child $\beta(B_k)$ of $\chi(c')$ (\Cref{fi:angle-constraint-a}). In this case, we need to impose on $c'$ a \emph{reflex-angle constraint}, requiring a reflex angle (270$^\circ$) at $c'$ in the orthogonal representation of $B_j$; it guarantees that we can glue the representations of $B_j$ and of $B_k$ at $c'$.  $(ii)$ In all other cases no constraint is needed for $c'$.

\item 
If $\deg(c|B_j) = 1$, then $B_j$ is a trivial block (i.e., a single edge) and no constraints are needed for $c$. 

\item
If $\deg(c|B_j) \geq 2$, let $B_k$ be the block such that $\beta(B_k)$ is the parent of $\chi(c)$: 
	$(i)$ If $\deg(c)=4$ and $\deg(c|B_k)=\deg(c|B_j)=2$ (\Cref{fi:angle-constraint-b}), then we impose an \emph{external reflex-angle constraint} on $c$, which forces $c$ to have a reflex angle on the external face $f$; block $B_k$ will be embedded on $f$.
	$(ii)$ If $\deg(c)=4$ with $\deg(c|B_k)=1$ (i.e., $B_k$ is a trivial block) and $\deg(c|B_j)=2$, then $\beta(B_j)$ has a sibling $\beta(B_h)$, where $B_h$ is a trivial block (\Cref{fi:angle-constraint-c}). In this case, we impose an \emph{external non-right-angle constraint} on $c$, which forces $c$ to have an angle larger than $90^\circ$ (either a flat or a reflex angle) on the external face~$f$; block $B_k$ will be embedded on~$f$, while $B_h$ will be embedded either on~$f$ (if $c$ will have a reflex angle in~$f$) or on the other face of $B_j$ incident to $c$ (if~$c$ will have a flat angle in~$f$).
	$(iii)$ If $\deg(c)=4$ with $\deg(c|B_k)=1$ and $\deg(c|B_j)=3$ (\Cref{fi:angle-constraint-d}), we impose an \emph{external flat-angle constraint} on~$c$, which forces $c$ to have its unique flat angle in the external face~$f$; again, $B_k$ will be embedded on~$f$.
	$(iv)$ If $\deg(c)=3$ and $\deg(c|B_j)=2$ (hence $\deg(c|B_k)=1$), we impose an external non-right-angle constraint on $c$, as in case~$(ii)$.
\end{itemize}

\begin{figure}[h]
    \centering
    \subfigure[]{\includegraphics[height=0.18\columnwidth,page=1]{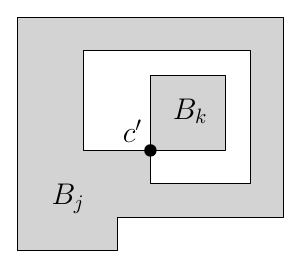}\label{fi:angle-constraint-a}}
    \hfil
    \subfigure[]{\includegraphics[height=0.18\columnwidth,page=2]{angle-constraints}\label{fi:angle-constraint-b}}
    \hfil
    \subfigure[]{\includegraphics[height=0.18\columnwidth,page=3]{angle-constraints}\label{fi:angle-constraint-c}}
    \hfil
    \subfigure[]{\includegraphics[height=0.18\columnwidth,page=4]{angle-constraints}\label{fi:angle-constraint-d}}
    \caption{Schematic illustration of cases that require angle constraints.}
 \label{fi:angle-constraints}
\end{figure}

\leconstrainedblock*
\begin{proof}
	A degree-2 vertex of $B_j$ corresponds either to a degree-2 vertex of $G$, or to a degree-3 or degree-4 cutvertex of $G$. Also, if $c$ is a cutvertex of $G$, either one of the blocks that contain $c$ has a degree-1 vertex or all blocks that contain $c$ have a cycle, which implies that, in any orthogonal representation of $G$, all these blocks have at least 3 reflex angles corresponding to either degree-2 vertices or bends. Hence, if $G$ admits a $b$-orthogonal representation with respect to ${\cal T}_{B_i}$,  
	the number of cutvertices of $G$ that belong to $B_j$ is at most $k+b$, which implies that the number of degree-2 vertices of $B_j$ is at most $p' = 2k + b$. Therefore, if $B_j$ contains more than $p'$ degree-2 vertices, we immediately conclude that the test is negative, as $B_j$ cannot admit a $b$-constrained orthogonal representation. 
	
    Suppose then that $B_j$ contains at most $p'$ degree-2 vertices. To continue the test, we apply on $B_j$ the same dynamic programming algorithm of \Cref{th:biconnected} based on the rooted SPQ$^*$R-trees of $B_j$, but we specialize this algorithm to take into account the required angle constraints at the cutvertices when every rooted SPQ$^*$R-trees of $B_j$ is visited. In the following we describe how to handle each type of constraint during a visit of an SPQ$^*$R-tree, where again $c$ denotes the cutvertex that must lie on the external face of $B_j$ (if $B_j \neq B_i$) and $c'$ denotes any other cutvertex of $G$ that belongs to~$B_j$. Note that, if $B_j \neq B_i$, we only need to consider those SPQ$^*$R-trees of $B_j$ rooted at reference chains that contain $c$, as we must restrict to those planar embeddings of $B_j$ with $c$ on the external face. 
	
	\begin{itemize}
		\item {\em Reflex-angle constraint at $c'$}. In this case $c'$ has degree two in $B_j$, thus $c'$ belongs to a maximal chain of edges represented by a Q$^*$-node $\nu$. It is enough to consider those spirality values $\sigma_\nu$ of $G_\nu$ that are compatible with a left/right turn at $c'$ and, for such a value of spirality, consider a representative $H_\nu$ having a left/right turn at $c'$.
		
		\item {\em External reflex-angle constraint at $c$}. Also in this case $c$ has degree two in $B_j$, and $c$ belongs to a maximal chain of edges represented by a Q$^*$-node $\rho$. In this case the SPQ$^*$-tree rooted at $\rho$ is the only one that we need to visit, thus $c$ belongs to the reference chain and is distinct from the poles $u$ and $v$ of $\rho$. Again, when we evaluate the condition of \Cref{le:chain} at the root level, it suffices to restrict to those representations $H^{uv}_\rho$ with a left turn at $c$.
		
		\item {\em External flat-angle constraint at $c$}. In this case $c$ has degree three in $B_j$, and we visit three possible SPQ$^*$R-trees, one for each of the three maximal chains having $c$ as one of its two end-vertices. In each of these SPQ$^*$R-trees, the root child $\nu$ is either a P-node of degree two or an R-node, and $c$ is one of the two poles of $\nu$. If $\nu$ is a P-node, for each target spirality value $\sigma_\nu$, one can just consider the subset of combinations of the angles at $c$ that guarantee a flat angle at the pole $c$ in the external face of $G_\nu$  when construct $H_\nu$ (see~\cite{DBLP:conf/gd/DidimoKLO22} for details). If $\nu$ is an R-node, it is enough to force a flat angle at $c$ on the external face of the rigid component in the test of \Cref{le:R-computation}; this is done by fixing the value of the flow on an arc of the flow-network.
		
		\item {\em External non-right-angle constraint at $c$}. In this case $c$ has degree two in~$B_j$. Hence, similarly to the external reflex-angle constraint, $c$ belongs to a maximal chain of edges represented by a Q$^*$-node $\rho$, and we only need to visit the SPQ$^*$R-tree rooted at $\rho$. When we evaluate the condition of \Cref{le:chain} at the root level, it is enough to only consider those representations $H^{uv}_\rho$ that have either a left turn at $c$ or no turn at $c$.
	\end{itemize}  
	
	All the described modifications of the testing algorithm in \Cref{th:biconnected} still take $O(p')$ time. Since $p' = O(k+b)$, the statement follows. 
\end{proof}

\thgeneralcase*
\begin{proof}
Suppose first that $G$ is 1-connected. Let $B_1, \dots, B_q$ be the blocks of $G$. For any given rooted BC-tree ${\cal T}_{B_i}$, we test whether $G$ admits a $b$-orthogonal representation with respect to ${\cal T}_{B_i}$. This requires to test for each block $B_j$ whether $B_j$ admits a $b$-constrained representation. This test is done with the algorithm of \Cref{le:constrained-block} and, if it is positive, the same algorithm computes one of these representations, call it $H_j$, having the minimum number of bends; denote by $b_j$ the number of bends of $H_j$. If $\sum_{j=1, \dots, q}b_j \leq b$ then the whole test is positive, and we can obtain a $b$-constrained representation of $G$ by merging all the $H_j$ at the cutvertices. Otherwise, we discard ${\cal T}_{B_i}$. We can execute this test for all rooted BC-trees, and return the $b$-orthogonal representation of $G$ having the minimum number of bends among those for which the test was positive (if any). If the test for each rooted BC-tree was negative, then there is no a $b$-orthogonal representation of $G$.  
Since $q=O(n)$, the time complexity immediately follows by \Cref{le:constrained-block}. 

Finally, if $G$ is not connected we apply the above procedure to each connected component of $G$ and test if the sum of the number of bends in the components exceeds $b$. Since the total number of blocks of $G$ is $O(n)$, by an analogous reasoning as above we obtain the same bound on the time complexity.
%
\end{proof}

\end{document}